\documentclass[11pt]{article}

\usepackage{fullpage}
\usepackage[utf8]{inputenc}
\usepackage{amsmath,amsthm}
\usepackage{graphicx}
\usepackage{xspace}
\usepackage{nicefrac}
\usepackage{url}
\newtheorem{theorem}{Theorem}
\newtheorem{lemma}[theorem]{Lemma}
\newtheorem{definition}{Definition}

\newcommand{\bid}{\text{bid}\xspace}

\newcommand{\optbidder}{i^*} 
\newcommand{\otherbidder}{i^o}   
\newcommand{\algbidder}{\textsc{algb}\xspace}
\newcommand{\ctr}{\mbox{ctr}\xspace}
\newcommand{\cpc}{\mbox{cpc}\xspace}
\newcommand{\cw}{\textsc{LW}\xspace}
\newcommand{\spend}{\text{spend}}
\newcommand{\eq}{\textsc{EQ}\xspace}

\newcommand{\opt}{\textsc{OPT}\xspace}

\newcommand{\optbj}{\bid(\optbidder (j), j)}
\newcommand{\otherbj}{\bid(\otherbidder (j), j)}
\newcommand{\randauction}{\text{\sc{Rand}}}

\newcommand{\aranyakignore}[1]{}

\title{Auction Design in an Auto-bidding Setting: Randomization Improves Efficiency Beyond VCG}
\author{Aranyak Mehta \thanks{Google Research, Mountain View, CA, \texttt{aranyak@google.com}}}

\date{}

\begin{document}
\maketitle

\begin{abstract}
Auto-bidding is an area of increasing importance in the domain of online advertising. We study the problem of designing auctions in an auto-bidding setting with the goal of maximizing welfare at system equilibrium. Previous results showed that the price of anarchy (PoA) under VCG is 2 and also that this is tight even with two bidders. This raises an interesting question as to whether VCG yields the best efficiency in this setting, or whether the PoA can be improved upon. We present a prior-free randomized auction in which the PoA is approx. 1.896 for the case of two bidders, proving that one can achieve an efficiency strictly better than that under VCG in this setting. We also provide a stark impossibility result for the problem in general  as the number of bidders increases -- we show that no (randomized) anonymous truthful auction can have a PoA strictly better than 2 asymptotically as the number of bidders per query increases. While it was shown in previous work that one can improve on the PoA of 2 if the auction is allowed to use the bidder's values for the queries in addition to the bidder's bids, we note that our randomized auction is prior-free and does not use such additional information; our impossibility result also applies to auctions without additional value information.
\end{abstract}

\section{Introduction}
Auto-bidding is an area of increasing importance in the online advertising ecosystem (see, e.g.,~\cite{googleautobiddingsupport,fbautobiddingsupport}, with significant innovation and adoption in recent years. With auto-bidding, each advertiser states its goals and constraints to an auto-bidding agent, which then converts those into per-auction bids. In a setting where all advertisers use auto-bidding, there one agent per advertiser, bidding optimally with respect to the other agents. We are then interested in the properties of the system equilibrium.

A prototypical auto-bidding setting is that of target-cost-per-acquisition (tCPA) in which the advertiser aims to maximize their conversions (sales after an ad-click) subject to an ROI constraint that the average cost-per-conversion is no bigger than a given target. Target-return-on-ad-spend (tROAS) generalizes tCPA to take into account the value of a conversion as well. Recent research captures the theoretical problems in this setting. In \cite{AggarwalBM19}, the authors introduced the problem of auto-bidding under advertiser goals and constraints and presented two results: First, they presented a bidding formula to bid optimally into truthful auctions. Secondly, they proved a price of anarchy result: If all advertisers adopt such optimal auto-bidding to bid into a (per-query) VCG auction, then the welfare at equilibrium is at least a half of the value obtained an an optimal allocation. They also presented an instance of the problem (with two bidders) in which there exists an equilibrium with welfare equal to a half of the optimal, showing that their analysis is tight. Thus, they showed that the price of anarchy (PoA) under VCG is 2. Further research~\cite{DengMMZ21} improves on the PoA of 2 by allowing the auction access to the values of the advertisers. For example, they show that by adding a boost to each bidder's bid equal to $c$ times the value for the ad, and then running a second-price auction yields a PoA of $\frac{2+c}{1+c}$.

In this paper, we investigate the design of auctions to improve on the PoA in the setting introduced in~\cite{AggarwalBM19}. As opposed to the direction taken in~\cite{DengMMZ21}, we consider the setting where the auction does not have access to any other information (e.g., the values) besides the bid. This restriction has practical motivation: the auction may not have access to the values of each impression, e.g., the bidding system which converts the values, objectives, and constraints into bids may be a separate system, or could even be owned by an third party different from the auction platform. Also, as opposed to the Bayesian setting, we do not assume that the auction has access to the distribution of the values or targets, or indeed if there is such an underlying distribution -- i.e., our auction is {\em prior-free}. We ask the question: \emph{Is it possible to improve over VCG in the prior-free setting without any additional information besides the bids?}

\subsection{Results}
We first consider the case of two bidders, and introduce a simple prior-free randomized truthful auction $\randauction(\alpha, p)$ (parameterized by a constant $\alpha \geq 1$, and probability $p \in (0, 1/2]$) and show ({\bf Theorem~\ref{thm:poa}}) that with a choice of $p = \frac{2}{5}$ and $\alpha \in [1, \frac{1 + \sqrt{85}}{6}]$, the auction has a price of anarchy no worse than $\frac{2\alpha + 6 + \frac{2}{\alpha}}{2\alpha + 1 + \frac{2}{\alpha}}$. Choosing $\alpha^* = \frac{1 + \sqrt{85}}{6} \simeq 1.703$, this shows that the PoA of $\randauction(\alpha^*, \tfrac{2}{5})$ is approximately $1.896$.
 We also show ({\bf Theorem~\ref{thm:poa-tight-example}}) that our analysis is tight for this range of $\alpha$, i.e., there is an instance in which the welfare for $\randauction(\alpha, \tfrac{2}{5})$ at equilibrium is precisely $\frac{2\alpha + 6 + \frac{2}{\alpha}}{2\alpha + 1 + \frac{2}{\alpha}}$. We note that when $\alpha = 1$ then the auction becomes identical to the second-price auction, and the PoA bound reduces to $2$, which means that our result strictly generalizes the result in~\cite{AggarwalBM19}.

Thus we have a prior-free auction with equilibrium welfare strictly better than that of VCG (\cite{AggarwalBM19} provides a tight example with two bidders for the PoA of VCG). This is surprising, and highlights the difference between the ROI constrained auto-bidding setting and the standard quasi-linear setting, in which VCG achieves maximum welfare. Further, similar results in the quasi-linear setting for improving over the revenue of VCG depend on either prior knowledge of the distribution of values (\emph{Bayesian auction design}) ~\cite{Myerson81}, or on evaluating over instances in which the values are picked from some unknown distribution from a well-formed class like regular distributions (\emph{prior-independent auction design}, e.g.,~\cite{DhangwatnotaiRY15}). In comparison, our result is in the prior-free setting, the auction has no other input besides the bids and is evaluated on any set of values, yet improves on the welfare of VCG in equilibrium. 
We note that the welfare objective in our setting is Liquid welfare which captures both the desire and ability to pay.
Even though the utility function is non-standard, the questions of incentive compatibility and mechanism design remain important: As shown in~\cite{AggarwalBM19}, it is only in a truthful auction that there is a simple (uniform) optimal bidding formula. We also restrict our attention to truthful auctions.

Our result above shows a strict improvement over the PoA of VCG in the two bidder setting. One can therefore ask for an auction which can get a PoA better than 2 more generally. While it is possible to extend the analysis in Theorem~\ref{thm:poa} to a fixed number of bidders and retain a PoA strictly less than 2 (we do not include this slightly generalized analysis here), we complement our positive result with a stark impossibility result in this setting. We show ({\bf Theorem~\ref{thm:poa-multiple-half}}) that as the number of bidders increases (in particular the number of bidders bidding per query increases), then the PoA of any (randomized) truthful anonymous auction is asymptotically no better than 2. This result closes the problem introduced in~\cite{AggarwalBM19} for prior-free auctions (restricted to auctions without any additional information besides the bids).

Finally, we note that the model, auction design, and analytical results are in a stylistic and theoretical setting (see related work below), and is not meant to accurately represent all aspects and complexities of real-world practical implementations.

\subsection{Related Work}
\label{sec:related-work}
As mentioned above, the closest related work is that in~\cite{AggarwalBM19} and \cite{DengMMZ21}. The work in~\cite{AggarwalBM19} introduced the problem and provided the PoA analysis for VCG in a more general setting (with multiple ROI constraints); we restrict our analysis to the prototypical tCPA setting (but it applies directly to tROAS), and also to the single slot per query setting.
In the context of the results in~\cite{DengMMZ21}, it is interesting that we can also improve over the PoA of VCG without the additional value information. There has also been recent work on understanding the optimal mechanism design in Bayesian settings~\cite{GolrezaiLP21,BalseiroDMMZ21} (with the latter describing a few different models of the setting); in comparison our work does not consider advertiser incentives (e.g., to misreport the targets) but only the system response (as in~\cite{AggarwalBM19,DengMMZ21}).

The auction introduced here ($\randauction(\alpha, p)$) is reminiscent of auctions studied in a few previous results. We will define $\randauction(\alpha, p)$ in Sec.~\ref{sec:prelim}, but informally, it allocates to the higher bidder only if its bid is an $\alpha \geq 1$ factor greater than the other bidder; otherwise it allocates to both  bidders randomly, with a higher probability to the higher bidder.

In~\cite{FuILS15}, the authors introduced a prior-independent auction called $(\epsilon, \delta)$-inflated second-price auction, where the parameter $\delta \geq 1$ is similar to our $\alpha$ parameter, but used differently: with probability $\epsilon$ it runs a second-price auction, and with the rest of the probability the highest bidder wins only if its bid is greater than the next highest bid by a factor of $\delta \geq 1$ (otherwise the item is unallocated). For this auction, it was proved that it achieves a fraction strictly better than $\frac{n-1}{n}$ of the optimal revenue (in the standard quasi-linear utility setting, where $n$ is the number of bidders), and an improved ratio of 0.512 was provided for two bidders. This result was generalized in~\cite{AllouahB20} which introduced a family of prior-independent auctions called threshold-auctions, and proved stronger results for revenue in the prior-independent setting with a focus on the two bidder setting. In fact, $\randauction(\alpha, p)$ lies in the family of threshold-auctions. 
One may expect that other similar members of threshold-auctions family may also improve over the PoA of VCG in the auto-bidding setting -- although perhaps not the specific auctions studied in~\cite{FuILS15,AllouahB20}. We note that a similar auction called the ratio-auction was described in~\cite{hartline2014optimal} in a different context.

The definition of liquid welfare used in the auto-bidding setting \cite{AggarwalBM19} is based on the liquid welfare definition introduced for the budgeted setting in~\cite{DobzinskiL14} (see also~\cite{AzarFGR17}). 
 
Finally, we note that prior to ROI based auto-bidding products such as tCPA and tROAS, a well-established and well-studied model was that of budgeted optimization -- how to bid under simple budget constraints. There are several lines of work in this model, e.g.,~\cite{FeldmanMPS07,BalseiroG19} among others.

\section{Preliminaries}
\label{sec:prelim}

\subsection{Auto-bidding and tCPA / tROAS constraints}
In the general auto-bidding problem introduced in \cite{AggarwalBM19}, each advertiser has a goal and a set of constraints. We will restrict our attention to tCPA, in which the advertiser's goal is to maximize the number of conversions, (sales after the click on the ad) subject to a constraint which says that the average cost of the conversions is no more than an advertiser input target $T$. 
In other words, the goal in tCPA bidding is to maximize the conversions subject to the expected spend being at most $T$ times the expected conversions.Thus the problem for a single bidder $i$ is:
\begin{align}
    \text{Maximize } &\sum_{j\in Q} x_{ij} \ctr_{ij} v_{ij} \label{lp:tcpa}\\
    s.t. \sum_{j\in Q} x_{ij} ctr_{ij} cpc_{ij} &\leq T(i)\cdot \sum_{j\in Q} x_{ij} ctr_{ij} v_{ij}\notag\\
    \forall j\in Q: 0 &\leq x_{ij} \leq 1\notag
\end{align}
Here, $Q$ is the set of queries, the $x_{ij}$ are the decision variables as to whether advertiser $i$ should buy the slot on the $j$th query (we assume a single slot per query for simplicity), $\ctr_{ij}$ is the click-through-rate for an ad of advertiser $i$ on the $j$th query (the probability of a click given an impression),  $\cpc_{ij}$ is the cost-per-click of the ad (determined by the auction from the other advertisers' bids and hence not known in advance), and $T(i)$ is advertiser $i$'s input target CPA. Finally $v_{ij}$ is the value of a click on an ad for advertiser $i$ on the $j$th query. For tCPA, $v_{ij}$ is the conversion rate (the probability of a conversion given a click). This formulation easily extends to tROAS by taking $v_{ij}$ to be the total value after the click, e.g., value of a sale.\\

\noindent{\bf Bidding formula. }It was shown in~\cite{AggarwalBM19} that the bidder can bid according to a simple bidding formula based on the optimal value of the variables of the dual of LP~(\ref{lp:tcpa}). If the auction is truthful, and we have the correct values of the optimal dual variables, then this bidding formula is optimal, i.e., the bidder $i$ buys the optimal solution to LP~(\ref{lp:tcpa}). In the tCPA setting above, given the optimal dual variable $\gamma_i \geq 0$ for the tCPA constraint, the optimal bid for bidder $i$ takes the simple form
\begin{equation}
\label{eq:bid-mult}
    bid(i, j) = \mu_i \cdot v_{ij}, \mbox{~~~~  where } \mu_i := 1 + \frac{1}{\gamma_i}
\end{equation}
$\mu_i \geq 1$ is called the bid multiplier for the bidder $A$.\medskip

{\bf Equilibrium. } An instance of the problem is a set of queries $j \in Q$ and a set of bidders $i\in A$, each with its own (private) tCPA constant $T(i)$. Each $(i,j)$ pair has a $ctr_{ij}$ and value $v_{ij}$. We are interested in the allocation outcome when all advertisers bid using the optimal bidding formula (\ref{eq:bid-mult}). Note that each bidder's bids depend on the dual optimal variable for its LP (\ref{lp:tcpa}) in which the cpcs are determined in the auction based on the other bidders' bids. Thus we have to study the system at equilibrium. 
\begin{definition}
\label{def:equilibrium}
Consider a (possibly randomized) auction $\mathcal{A}$ and an instance $\mathcal{I}$ with set $A$ of bidders and $Q$ of queries. Let $v_{ij}$ be the value of bidder $i$ for query $j$ in the instance, and consider a set of bid-multipliers $\{\mu_i\}_{i\in A}$, so that $bid(i, j) := \mu_i v_{ij}$ is the bid of bidder $i$ for query $j$.
Let $\{ x_{ij} \}_{i\in A, j\in Q}$ be the (probabilistic) allocation achieved with bids $bid(i,j)$ under the auction $\mathcal{A}$, and let $p_{ij}$ be the per-unit price (cpc) charged to bidder $i$ for query $j$. For $\Delta, \gamma \geq 0$, the bids $\{bid(i, j)\}$ are said to be in a $(\Delta, \gamma)$-equilibrium if the following hold.
\begin{itemize}
    \item Each bidder satisfies its tCPA constraint up to a multiplicative factor of $\gamma$:
    $$\forall i\in A: \sum_{j\in Q} p_{ij}ctr_{ij}x_{ij} ~\leq~ \left(1 + \gamma\right) T(i) \sum_{j\in Q} v_{ij}ctr_{ij} x_{ij}$$
    \item No bidder can deviate from its bid-multiplier unilaterally and gain more than a $\Delta$ amount of additive value while still satisfying its tCPA constraint (up to $\gamma$). More precisely, suppose bidder $i$ changes its bid-multiplier to $\mu'_i$ and this changes its allocation to $\{x'_{ij}\}$ and the prices to $\{p'_{ij}\}$. Then $\forall i\in A$
    \begin{align*}
       \mbox{Either\ \ \ \  }\sum_j v_{ij}ctr_{ij} x'_{ij} ~&<~  \sum_j v_{ij} ctr_{ij}x_{ij} + \Delta\\&~~~~\mbox{(does not gain more than } \Delta),\\
        \mbox{or  \ \  \ \ } \sum_j p'_{ij}ctr_{ij} x'_{ij} ~&>~  \left(1 + \gamma\right)\sum_j v_{ij}ctr_{ij} x'_{ij}\\&\mbox{(violates the tCPA constraint by more than $\gamma$).}
    \end{align*}
\end{itemize}
\end{definition}

Our positive result on PoA (Sec.~\ref{sec:poa-2-bidders}) will use an exact definition of equilibrium (with $\Delta = \gamma =  0)$, while our impossibility result (Sec.~\ref{sec:many-bidders}) holds for any $\Delta, \gamma >0$.\medskip

{\bf Liquid Welfare. } Since the spend is constrained, we need to define welfare carefully. In this setting the appropriate definition of welfare is the Liquid Welfare (\cw), first defined in \cite{DobzinskiL14} for the budgeted allocation case, and generalized to the auto-bidding setting in \cite{AggarwalBM19}. Liquid welfare captures both the willingness and the ability to pay for a given allocation. For a general auto-bidding problem (with potentially multiple constraints on spend), the liquid welfare for a given allocation is defined \cite{AggarwalBM19} as the sum over bidders, of the minimum value of the right hand sides of the constraints of the bidder's LP, in that allocation. The tCPA LP (\ref{lp:tcpa}) has a single constraint, and the LW becomes equivalent to the sum of tCPA-weighted conversions: 
\begin{definition}
\label{def:lw}
In the tCPA setting,
\begin{equation}
\label{eq:def-lw}
\cw(\{x_{ij}\}) = \sum_{i \in A} T(i) \sum_{j\in Q} x_{ij}\ctr_{ij}v_{ij}
\end{equation}
where $A$ is the set of advertisers, $T(i)$ is the target CPA for advertiser $i$, $\{x_{ij}\}$ denotes the allocation, and $v_{ij}$ is the conversion rate of an ad of advertiser $i$ on query $j$. Thus the liquid welfare is the total number of tCPA-weighted conversions.
\end{definition}

Define $\cw(i, \{x_{ij}\})$ as the contribution of advertiser $i$ to the liquid welfare (Eq. \ref{eq:def-lw}), and let $\spend(i, \{x_{ij}\})$ denote the total expected spend of advertiser $i$ in the auction. We see from LP (\ref{lp:tcpa}) that the latter is a lower bound on the former.
\begin{equation}
    \label{eq:spend-leq-a-lw}
    \spend(i, \{x_{\cdot}\}) \leq LW(i, \{x_{\cdot}\})
\end{equation}

{\bf Price of Anarchy} is defined as
$$PoA = \max_{\mathcal{I}} \max_{eq \in Eq(\mathcal{I})} \tfrac{LW(OPT(\mathcal{I}))}{LW(eq)} $$
where $\mathcal{I}$ denotes an instance of the problem, $Eq(\mathcal{I})$ is the set of equilibria in instance $\mathcal{I}$, and $OPT(\mathcal{I})$ denotes the allocation with the highest liquid welfare for that instance.

\subsection{Auction $\randauction(\alpha, p)$}
We define auction $\randauction(\alpha, p)$ for two bidders parameterized by $\alpha \geq 1$ and $p\in (0, 1/2]$. As mentioned in Sec.~\ref{sec:related-work}, this is a member of the \emph{threshold-auctions} family from~\cite{AllouahB20}, and somewhat similar to the \emph{bid-inflation} auction from~\cite{FuILS15} and the ratio-auction from~\cite{hartline2014optimal}.
\begin{definition}
\label{def:rand}
Let the bids be $b_1 , b_2$, and without loss of generality assume that $b_1 \geq b_2$. Then the allocation function is as follows:
\begin{itemize}
    \item If $b_1 \geq \alpha\cdot b_2$, allocate to bidder 1.
    \item Else, allocate to bidder 1 w.p. $1-p$, and to bidder 2 w.p. $p$.
\end{itemize}
The truthful prices fall out of the allocation function in the natural manner:
\begin{align*}
    &\mbox{If } b_1 \geq \alpha\cdot b_2, \mbox{ then } cost_1 = b_2\left(\frac{p}{\alpha} + (1-2p) + p\cdot\alpha\right) \\
    &\mbox{Else, } E[cost_1] = b_2 \left(\frac{p}{\alpha} + (1-2p)\right), \mbox{ and } E[cost_2] = \frac{p\cdot b_1}{\alpha}
\end{align*}
\end{definition}

Note that $\randauction(1, p)$ is identical to a second-price auction for any $p$.

\section{An improved price of anarchy for the two bidders case}
\label{sec:poa-2-bidders}
In this section we prove our main positive result, that the price of anarchy of $\randauction(\alpha, p)$ in the two bidder setting, is strictly better than 2 for a choice of $\alpha$ and $p$, and hence is a strict improvement over VCG.
\subsection{Techniques and Intuition}
The proof for the PoA of 2 for VCG in~\cite{AggarwalBM19} proceeds at a high level as follows: For any query, let the opt-bidder be the one that is allocated the query in the optimal allocation. Now, either (case 1) the opt-bidder got allocated in the equilibrium -- in which case the optimal value is obtained for that query, or (case 2) another bidder was allocated -- in which case its spend on the query is at least the opt-bidder's bid, since that is a floor in the second price auction. Due to the ROI constraints, the total spend of an advertiser is also a lower bound on its contribution to liquid welfare, and therefore the spends in queries in case 2 help bound the welfare by the optimal welfare for those queries. Putting these two cases together, we get a ratio of 2 for VCG.

The trade-offs are different in $\randauction(\alpha, p)$ (for $\alpha > 1$). On the positive side, if the auction allocates to the incorrect bidder, and if its bid is greater than that of the opt-bidder by more than a factor of $\alpha$, then we see a gain in the attributions in case 2 above, since the spend is strictly greater than the opt-bidder's bid (see Def.~\ref{def:rand}). However, this also comes with potential losses in attribution: When the opt-bidder is the highest but the second-highest bid is close to its bid then we only select the opt-bidder w.p. $1-p < 1$, and the spend is also lower than the second-price auction. A similar dynamic plays out when the opt-bidder is not the highest, but is close to the highest, in which case you allocate to the opt-bidder w.p. $p$, but the spend is again lower than in the second-price auction. 

With such opposing cases, one would expect to gain in the aggregate only if $\alpha$ and $p$ are well-tuned to the underlying values -- e.g., in Bayesian setting when the value distributions are known, or perhaps even in the prior-independent setting when there are underlying regular or MHR type distributions. Surprisingly, we show that for a range of $\alpha$ and $p$, the tradeoffs \emph{always} work out in our favor, for any fixed set of values, i.e., in a prior-free manner.

\subsection{Bounding the PoA}

\begin{theorem}
\label{thm:poa}
For $p = \tfrac{2}{5}$, and $\alpha \in [1, \frac{1 + \sqrt{85}}{6}]$, the PoA of $\randauction(\alpha, p)$ (defined for two bidders) is at most $\frac{2\alpha + 6 + \frac{2}{\alpha}}{2\alpha + 1 + \frac{2}{\alpha}}$.
In particular, for $\alpha^*=\frac{1 + \sqrt{85}}{6} \simeq 1.703$, the PoA of $\randauction(\alpha^*, \frac{2}{5})$ is at most a value $\simeq 1.896$.
\end{theorem}
\begin{proof}
Consider any instance $\mathcal{I}$ consisting of two bidders, and a set of queries $Q$. Bidder $i \in \{0,1\}$ has a tCPA of $T(i)$. The value of bidder $i$'s ad on query $j \in Q$ is $v(i, j)$. For ease of exposition, we will assume that all click-through rates are 1.

Fix an optimal allocation \opt, and any equilibrium allocation \eq. For $j\in Q$, let $\optbidder(j)$ denote the bidder to which $\opt$ allocates $j$, and let $\otherbidder(j)$ denote the other bidder. Let $\algbidder(j)$ be the bidder to which the auction allocates $j$ in \eq.

For any $j\in Q$, define $OPT(j)$ as the contribution of $j$ to the liquid welfare in the optimal allocation, i.e. (from Def.~\ref{def:lw}),
$$OPT(j) = T(i^*(j)) \cdot v(i^*(j), j)$$
Similarly, for any subset $Q' \subseteq Q$, define $OPT(Q')$ as the contribution of $Q'$ to the liquid welfare in $OPT$:
$$OPT(Q') = \sum_{j \in Q'} OPT(j)$$

Firstly, note that, for all $j \in Q$, using the optimal bidding formula (Eq.\ref{eq:bid-mult}), and the fact that the optimal bid multiplier for each bidder is at least 1, we get
\begin{align}
 \optbj &= \mu(\optbidder)\cdot T(\optbidder(j))\cdot v(\optbidder(j), j) \notag\\
 &\geq T(\optbidder(j))\cdot v(\optbidder(j), j) \notag\\
 &= \opt (j)   \label{eq:bid_ge_opt}
\end{align}

{\bf Query partition: } We partition $Q$ into four parts based on the relative values of the bids of $\optbidder(j)$ and $\otherbidder(j)$. For a query $j$ in each part, we will compute, using the definition (\ref{def:rand}) of $\randauction(\alpha, p)$:\\
\ \ (a) the probability that $j$ is allocated to $\optbidder(j)$, and\\ \ \ (b) the expected spend on query $j$.

\begin{enumerate}
\item Let $Q_1$ be the set of queries $j$ such that $$\optbj \leq  \frac{\otherbj}{\alpha}$$
For each $j \in Q_1$, $\randauction(\alpha, p)$ allocates $j$ to $\otherbidder(j)$ w.p. 1:
$$Pr[\algbidder(j) = \optbidder(j)] = 0$$
Further, the expected spend on $j$ is 
\begin{align*}
E[\spend(j)] &= \left(p\cdot \alpha + (1-2p) + \frac{p}{\alpha}\right)\cdot \optbj \\
&\geq~ \left(p\cdot \alpha + (1-2p) + \frac{p}{\alpha}\right)\cdot\opt(j)\ \ \  \mbox{(using Eq.\ref{eq:bid_ge_opt})}  
\end{align*}

\item Let $Q_2$ be the set of queries $j$ such that $$\frac{\otherbj}{\alpha} \leq \optbj \leq \otherbj$$
For each $j \in Q_2$: $$Pr[\algbidder(j) = \optbidder(j)] = p$$
The total expected spend on $j$ is the sum of the expected spends for the two bidders which are:
\begin{align*}
E[\spend(\otherbidder(j), j] &= p \cdot\frac{\optbj}{\alpha} + (1-2p)\cdot \optbj,\\
E[\spend(\optbidder(j), j)] &= p\cdot  \frac{\otherbj}{\alpha}~ \geq~ p\cdot  \frac{\optbj}{\alpha},
\end{align*}
where the last inequality is because $\otherbj \geq \optbj$ in this case.
Thus, using Eq.\ref{eq:bid_ge_opt}, we get 
\begin{align*}
E[\spend(j)] \geq \left(\frac{2p}{\alpha} + (1-2p)\right)\cdot \opt(j)
\end{align*}

\item Let $Q_3$ be the set of queries $j$ such that $$\otherbj \leq \optbj \leq \alpha \cdot \otherbj$$
Then for each $j \in Q_3$: $$Pr[\algbidder(j) = \optbidder(j)] = 1-p$$
The total expected spend is the sum of spends of the two bidders which are:
$$E[\spend(\otherbidder(j), j] = p\cdot \frac{\optbj}{\alpha}$$
and
\begin{align*}
E[\spend(\optbidder(j), j)] &=  \left(\frac{p\cdot\otherbj}{\alpha} + (1-2p)\cdot\otherbj\right)\\
&\geq \left(\frac{p\cdot\optbj}{\alpha^2} + (1-2p)\cdot\frac{\optbj}{\alpha}\right)
\end{align*}
where the last inequality is because $\otherbj \geq \frac{\optbj}{\alpha}$ in this case.
Thus, using Eq.\ref{eq:bid_ge_opt}, we get 
\begin{align*}
E[\spend(j)] ~\geq~ \left(\frac{1-p}{\alpha} + \frac{p}{\alpha^2}\right)\cdot \opt(j)
\end{align*}

\item Let $Q_4$ be the set of queries $j$ such that $$\optbj \geq  \alpha\cdot \otherbj$$
Then for each $j \in Q_2$: $$Pr[\algbidder(j) = \optbidder(j)] = 1$$ while the expected spend can be 0 (if $\otherbj = 0$). 

\end{enumerate}
\vspace{1em}
\noindent{\bf Bounding the welfare:} 
Define the following constants:  
\begin{align*}
    &m_1 = 0 \ \ \ \ \ \ \ \ \ \ \  &s_1 &= p\cdot \alpha + (1-2p) + \frac{p}{\alpha}\\
    &m_2 = p      &s_2 &= \frac{2p}{\alpha} + (1-2p)\\
    &m_3 = 1-p      &s_3 &=  \frac{1-p}{\alpha} + \frac{p}{\alpha^2}\\
    &m_4 = 1        &s_4 &= 0
\end{align*}
From the calculations above, we have:
\begin{align}
    \forall k = 1..4&: \forall j \in Q_k, ~Pr[\algbidder(j) = \optbidder(j)] = m_k \label{eq:m-k-def}\\
    \forall k = 1..4&: \forall j \in Q_k, ~\spend(j) \geq s_k\cdot \opt(j)\label{eq:s-k-def}
\end{align}

Now, for each case $k = 1\ldots 4$, we get the following two bounds on the welfare obtained in equilibrium.

Firstly, considering the probability that $j$ is allocated to the optimal bidder $i^*$ in the equilibrium, we get, using Eq. \ref{eq:m-k-def}:
\begin{align}
    E[\cw(Eq)] &\geq \sum_{k=1}^4 \sum_{j\in Q_k}Pr\big[\algbidder(j) = \optbidder(j)\big]\cdot \opt(j)\notag\\
    &\geq \sum_{k=1}^4 m_k \cdot \opt(Q_k)     \label{eq:mi-bound}
\end{align}
Secondly, 
\begin{alignat}{2}
    E[\cw(Eq)] & = \sum_{i \in A} E[\cw(i, Eq)] &&\geq \sum_{i\in A} spend(i, Eq) \notag\\
    &=\sum_{k=1}^4 \sum_{j\in Q_k}E[\spend(j)] &&\geq \sum_{k=1}^4\sum_{j\in Q_k} s_k \cdot\opt(j)\notag\\
    &= \sum_{k=1}^4 s_k \opt(Q_k)  &&  \label{eq:si-bound}
\end{alignat}
where the first inequality follows from using Eq.~\ref{eq:spend-leq-a-lw} and the second inequality follows from Eq.~\ref{eq:s-k-def}.\medskip

\noindent{\bf Bounding the PoA via a factor revealing LP. } Define the variables $x_k := \opt(Q_k)$, for $k = 1\ldots 4$. Normalizing $OPT(Q) = 1$, we get 
\begin{equation}
    \label{eq:sumxk}
    \sum_{k} x_k = 1
\end{equation}

We can now bound the overall price of anarchy by minimizing $E[\cw(Q)]$ over the constraints \eqref{eq:mi-bound}, \eqref{eq:si-bound}, \eqref{eq:sumxk}. Let $z$ denote $E[\cw(Q)]$, then we have the following linear program and its dual LP (with variables $\gamma, \beta, \delta$):

\begin{minipage}{0.5\textwidth}
\begin{align*}
    \mbox{Minimize } z \\
    s.t.~~ z - \sum_k m_k x_k &\geq 0 \\
    z - \sum_k s_k x_k &\geq 0 \\
    \sum_k x_k &\geq 1 \\
    \forall k, x_k &\geq 0\\
\end{align*}
\end{minipage}
\vline\hfill
\begin{minipage}{0.5\textwidth}
\begin{align*}
    \mbox{Maximize  } \delta \\
    s.t.~~\forall k: ~ -m_k \gamma - s_k \beta + \delta~
    &\leq~ 0 \\
    \gamma + \beta &\leq 1 \\
    \gamma, \beta, \delta &\geq 0\\ \\ \\
\end{align*}
\end{minipage}

Consider the following dual solution:
\begin{equation}
    \label{eq:dual-feasible}
    \delta = \frac{1}{\left(1 + \frac{1}{s_1}\right)}, \ \ \ \gamma = \delta, \ \ \ \beta = \frac{\delta}{s_1}
\end{equation}
We show that this is a feasible solution for the choice of\footnote{This choice of $p=\tfrac{2}{5}$ is found by a search and is almost optimal. We note that an arguably more natural choice of $p=\tfrac{1}{3}$ (which gives an even increase in allocation probability as the bid increases) is only slightly worse, giving a PoA of approx. 1.91.} $p = \tfrac{2}{5}$ and a range of $\alpha$. For $k = 1$ and $k = 4$, the constraints in the dual LP are satisfied and in fact tight, since $m_1 =0, m_4 = 1,$ and $s_4 = 0$. Note also that 
$$\gamma + \beta = \delta \left(1 + \frac{1}{s_1}\right) = 1$$
so the last constraint is satisfied and tight as well.

For $k = 2$, we need:
$-m_2\gamma - s_2 \beta + \delta \leq 0$.
Plugging in the choice of $p=\tfrac{2}{5}$ and the values of $m_2, s_2, \gamma, \beta, \delta$, this reduces to:
$$3\alpha^2 - \alpha - 7 \leq 0,$$
which is satisfied for $\alpha \in [1, \frac{1 + \sqrt{85}}{6}]$.  Note that $\frac{1 + \sqrt{85}}{6} \simeq 1.703$.

For $k = 3$, we need:
$-m_3\gamma - s_3 \beta + \delta \leq 0,$
which reduces to:
$$4\alpha^3 + 2\alpha^2 - 11\alpha - 10 \leq 0 $$
This is satisfied for $\alpha \in [1, r]$, with $r \simeq 1.8$.

Thus we see that the solution \eqref{eq:dual-feasible} is a feasible solution for the dual LP, for $\alpha \in [1, \frac{1 + \sqrt{85}}{6}]$.

The value of the dual objective is 
$\delta = \frac{1}{\left(1 + \frac{1}{s_1}\right)} = \frac{2\alpha + 1 + \frac{2}{\alpha}}{2\alpha + 6 + \frac{2}{\alpha}}$. This means that the minimum value of the primal LP, which is a lower bound on the value obtained in the auction equilibrium, is no less than this value. Since $\opt$ is normalized to 1, this completes the proof the theorem. For $\alpha = \frac{1 + \sqrt{85}}{6}$, the value of the dual objective is approx. $0.527$, i.e., a price of anarchy of approx. $1.896$. 

For completeness, we note that there is a primal solution with the same value as this dual solution: $x_4 = z = \delta$, $x_1 = \frac{\delta}{s_1}$, and $x_2 = x_3 = 0$.
\end{proof}

\subsection{A tight example}
One may notice that the bound obtained above does not fully use the properties of the equilibrium -- it does use the fact that tCPA constraints are satisfied (in Eq.~\ref{eq:si-bound}), but does not use the property that no bidder has an incentive to change its bid-multiplier. So we may expect the bound to be very loose. Surprisingly, we show that there exists an instance in which the bound is tight, and the values of $x_k := OPT(Q_k)$ in an equilibrium in the instance are precisely those corresponding to an optimal LP solution in the proof above.

\begin{figure}[ht]
	\centering
  		\includegraphics[width=0.4\textwidth]{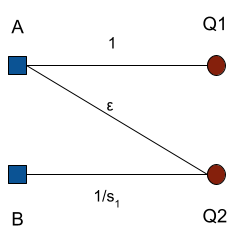}
  		\caption{A tight example.}
  		\label{fig:tight-example}
\end{figure}

\begin{theorem}
\label{thm:poa-tight-example}
The analysis in Theorem~\ref{thm:poa} is tight, i.e., there exists an instance with an equilibrium under $\randauction(\alpha, p)$, in which the ratio of the optimal welfare to the welfare in the equilibrium is exactly $1+\tfrac{1}{s_1}$, where $s_1 = p\cdot \alpha + (1-2p) + \frac{p}{\alpha}$. In particular, for $p=\tfrac{2}{5}$, the ratio is exactly $ \frac{2\alpha + 6 + \frac{2}{\alpha}}{2\alpha + 1 + \frac{2}{\alpha}}$.
\end{theorem}
\begin{proof}
Consider the instance in Fig.~\ref{fig:tight-example}. This instance is similar to the one used in~\cite{AggarwalBM19} except for the value on edge $(B, 2)$ (reduced from 1 to $1/s_1$).
We have two advertisers  $a, b$ and two queries $1,2$. The values for advertiser $a$ are $v_1^{a}=1$ and $v_2^{a}=\epsilon$ and for advertiser $b$ are $v_1^{b}=0$ and $v_2^{b}=\frac{1}{s_1}$. The click-through rates for all ad-query pairs are set to 1. For both the advertisers, we have a tCPA constraint with both $tCPA = 1$, so the constraint is that the spend should be at most the value.

Consider the bid multipliers: 
\begin{equation}
    \label{eq:tight-example-bms}
    \mu_A = \frac{\alpha}{\epsilon s_1} + 1, ~~\mu_B = 1
\end{equation}
We will first show that this is an equilibrium:

\begin{itemize}
\item For item 2, $A$'s bid is $\frac{\alpha}{s_1} + \epsilon$, while $B$'s bid is $\frac{1}{s_1}$, so $A$ wins item 2 with probability 1, and pays $s_1$ times $B$'s bid, equal to 1. $A$ also gets item 1 for free since $B$'s bid is 0. Thus $A$'s spend is $1$, which is less than its value, $1 + \epsilon$, for the two items obtained. Thus $A$'s tCPA constraint is satisfied, and since it gets both the items w.p. 1, it has no incentive to deviate from its bid of $\mu_A$.

\item $B$ gets nothing, but it can not bid up (increase $\mu_B$) to get item 2 with any probability and still satisfy its tCPA constraint: 

\begin{itemize}
\item If it bids very high to obtain item 2 with probability 1, then it has to pay $s_1$ times $A$'s bid which becomes $\alpha + \epsilon s_1$ which is more than the value of $\frac{1}{s_1}$ since $\alpha \geq 1$ and $s_1 \geq 1$. 

\item If it bids up to win the item with probability $1-p$, then its spend is 
$bid(A)\left( \tfrac{p}{\alpha} + (1-2p) \right) > \tfrac{\alpha}{s_1}\left( \tfrac{p}{\alpha} + (1-2p) \right)$. But this is at least the expected value for $B$ which is $\frac{1-p}{s_1}$ (since $\alpha \geq 1)$.

\item If $B$ bids up to win item 2 with probability $p$, then its spend is 
$p\cdot \tfrac{bid(A)}{\alpha} > \tfrac{p}{s_1}$, which is the expected value for $B$.
\end{itemize}
Thus $B$ does not want to defect from its bid of $\mu_B = 1$.
\end{itemize}
Thus the bid multipliers (\ref{eq:tight-example-bms}) are in equilibrium. In this equilibrium, the total value obtained is $1 + \epsilon$, while the optimal allocation (item 1 to $A$, item 2 to $B$) gets a value of $1 + \frac{1}{s_1}$, thus proving the theorem, as $\epsilon \rightarrow 0$.
\end{proof}

\section{An impossibility result with many bidders}
\label{sec:many-bidders}
In this section we will show an impossibility result when we have a large number of bidders bidding per query.

\subsection{Intuition}
Suppose we try to extend the proof in Sec.~\ref{sec:poa-2-bidders} to the case of many bidders per query (for some appropriate generalization of $\randauction$ to more than two bidders). The problem that arises is that the contributions to the welfare from queries in classes $Q_2$ and $Q_3$  can become very small. To see this, note that a query in type $Q_2$ or $Q_3$ has the opt-bidder among the highest bidders within some margin ($\alpha$). Now, the probability of the opt-bidder winning can be very low if there are a large number of bidders within the margin, since the anonymous auction has to fairly randomize between all of them. Thus, $m_2$ and $m_3$ can be close to 0. 
Now the bad scenario that could occur is that there is an instance and an equilibrium in which there are no queries in $Q_1$ ($x_1 = 0$), and a mix of queries, half in $Q_4$ (auction picked the opt-bidder) and half in $Q_2$ or $Q_3$. If such a scenario arises in equilibrium then the welfare in any such generalization of $\randauction$ would be no better than a half of the optimal welfare.

In fact, we show that such a scenario can be constructed in equilibrium for \emph{any} randomized auction. We do so by creating an instance for any given randomized auction using the intuition for $\randauction$ above. Queries will either be in $Q_4$ (i.e., the opt-bidder
wins without any competition), or will be such that there are a large number of bidders with almost equal bids, and the opt-bidder is the lowest among them. In that case the opt-bidder has a very low chance of winning, but the overall spend among the other bidders is also not high enough to increase the welfare. The technical difficulty here is that the space of randomized auctions is very large, so we have to construct an instance that relies on very basic auction properties so that every claim holds for all auctions.

\subsection{Anonymity and Max-Threshold}
\label{sec:1cs}
We start by defining two auction properties.
\begin{definition}
\label{def:anonymous}
A randomized auction is said to be \emph{anonymous} if its allocation does not depend on the identity of the bidders but only on the relative bid values. 
\end{definition}
Our proof of the impossibility result holds only for anonymous auctions.
Anonymity is often standard when deriving general results; note that anonymity precludes some known auctions such as those with personalized reserves (but these are not very applicable in the prior-free setting anyway). We will use the following:
\begin{lemma}
\label{lem:fairness-fraction}
In any anonymous auction, if there are $k$ bidders, 
then the lowest bidder wins the item with probability at most $\frac{1}{k}$.
\end{lemma}
\begin{proof}
Anonymity in a truthful auction (with a monotonic allocation function) implies fairness -- a higher bidder can not get a lower probability of allocation. Fairness implies the property required. 
\end{proof}

Next, since we want to prove the result for all (anonymous) randomized auctions, we need to identify a technical property of any auction that we will use in our construction. This is a technical requirement used to consider cases of auctions that do not sell the item fully to a bidder even if it is the only bidder and bids as high as needed; intuitively such auctions can only have worse efficiency.
\begin{definition}
\label{def:max-threshold}
Consider the setting when there is exactly one bidder for the item and puts a bid of $b$. As $b$ increases, the probability $P(b) = Pr[\mbox{win with bid }b]$ of winning the item must increase in a truthful auction. Define the \emph{max-probability}
$$\pi^* = \mathop{lim}_{b\rightarrow \infty} P(b)$$
be (the limit of) the highest probability with which the bidder can win. For simplicity of notation, we will assume that the limit is reached\footnote{For example, the standard assumption of consumer sovereignty implies that $\pi^* = 1$ and is achieved at a high enough bid.}. Define the \emph{max-threshold} as the lowest bid that achieves this highest probability.
$$M^* = min\{b: P(b) = \pi^*\}$$
\end{definition}
We will use the following:
\begin{lemma}
\label{lem:1cs-cost}
For any (randomized) truthful auction with a max-probability of $\pi^*$ and a max-threshold of $M^*$, the cost when there is a single bidder bidding $b \geq M^*$ is at most $\pi^* M^*$. 
\end{lemma}
\begin{proof}
This follows directly from Myerson's truthful pricing formula, where one can see that if the max-threshold bid is $M^*$ and the max-probability is $\pi^*$, then the area above the allocation curve (at a bid at least $M^*$) is at most $\pi^* \cdot M^*$. The bound is tight when the auction is a second price auction with a reserve price equal to $M^*$ (and $\pi^* = 1$).
\end{proof}

\subsection{Proving the bound}
\begin{theorem}
\label{thm:poa-multiple-half}
For any randomized truthful anonymous auction, for any $\Delta, \gamma > 0$, there exists an instance with sufficiently many bidders $n$ for which there exists a $(\Delta, \gamma)$-equilibrium in which the total value is asymptotically $\frac{1}{2}$ of the optimal value for the instance, as $n \rightarrow \infty$. 
\end{theorem}
\begin{proof}

Let the given the randomized auction $\mathcal{A}$ 
have max-probability $\pi^*$ and max-threshold $M^*$. For simplicity of notation going forward, if $M> 1$ then we re-scale all values in the following construction wlog so that $M = 1$. Note that we can assume $\pi^* \geq \frac{1}{2}$, otherwise we would be done, by taking an instance with one bidder and one item.\medskip

{\bf Constructing the instance. } Consider the following instance. 
There are $2k$ bidders in total: $k$ bidders $A_0,\ldots, A_{k-1}$, and $k$ bidders $B_0,\ldots, B_{k-1}$. All bidders have a tCPA of 1, i.e., they are constrained to have their total cost to be at most the total value achieved (up to a factor of $1+\gamma$). Correspondingly, there are $2k$ queries: $k$ queries $P_0,..,P_{k-1}$ and $k$ queries $Q_0, .., Q_{k-1}$. We pick $a \geq 1, V \geq 1$ to be fixed later.
\paragraph{Values for $P_i$:} For each $i \in [0, k-1]$, only $A_i$ has a non-zero value for $P_i$, with $v(A_i, P_i) = \frac{a\cdot V}{\pi^*}$.
\paragraph{Values for $Q_i$:} The values of the bidders for the $Q_i$ are in the following pattern. Firstly, for each $i \in [0, k-1]$, $v(B_i, Q_i) = V$, and $v(B_j, Q_i) = 0, \forall j \neq i$. Next, we  define the $k$-long tuple 
$$\tau^0 = (a\cdot V+\rho,\  a\cdot V+2\rho,\  \ldots, \ a\cdot V+k\rho)$$
Here $\rho >0$ is some very small constant used only for tie-breaking. For $i \in [1, k-1]$, define $\tau^i$ to be the $i^\emph{th}$ rotation of $\tau^0$, i.e., $\tau^i = (a\cdot V(i+1)\rho, \ a\cdot V+(i+2)\rho,\  ...,\  a\cdot V + k\rho,\  a\cdot V+\rho,\  a\cdot V+ 2\rho,\ \ldots,\  a\cdot V+i\rho)$. 

Now, for each $i \in [0, k - 1]$, the values of the bidders $\{A_j\}$ for $Q_i$ are determined as follows. Let $\epsilon > 0$ be a constant to be fixed later. The tuple of values $$\left(v(A_0, Q_i), \ v(A_1, Q_i),\  \ldots,\  v(A_{k-1}, Q_i)\right)$$ is set to be equal to $\epsilon \cdot \tau^i$, where the latter multiplication is coordinate-wise (so, e.g., $v(A_0, Q_i) = \epsilon(a\cdot V + (i+1)\rho)$). Fig.~\ref{fig:multiple-tight-example} shows the instance for $k=2$, i.e., with four bidders.\\

\begin{figure}[ht]
	\centering
  		\includegraphics[width=0.75\textwidth]{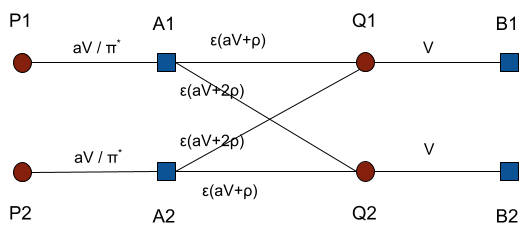}
  		\caption{An instance with $k=2$, i.e., 4 bidders and 4 queries. Here, $V > 1, a > 1, \epsilon > 0$. We take $V$ to be large, $a\rightarrow 1$, and $\epsilon\rightarrow 0$. $\rho$ can be considered a formal variable for tie-breaking.}
  		\label{fig:multiple-tight-example}
\end{figure}
{\bf Bid multipliers: }
For this instance, consider the following bid-multipliers: $\mu(A_i) = \frac{1}{\epsilon}, ~ \forall i$, and $\mu(B_i) = 1, ~ \forall i$.
\paragraph{Bids for $P_i$: } Thus, for each $i$, $bid(A_i, P_i) = \frac{a\cdot V}{\epsilon\cdot \pi^*}$ and $bid(X, P_i) = 0, \forall X \neq A_i$. 
\paragraph{Bids for $Q_i$: }
For each $i$, query $Q_i$ gets the following bids: Firstly $bid(B_i, Q_i) = V$, and $bid(B_j, Q_i) = 0, \forall j\neq i$.
Further, the tuple $$(bid(A_0, Q_i), bid(A_1, Q_i), ..., bid(A_{k-1}, Q_i)) = \tau^i.$$

{\bf The goal of the rest of the proof:} To prove the theorem, we need to prove that for this instance and these bid-multipliers:
\begin{enumerate}
    \item The welfare achieved in the auction with these bids is close to a half of the optimal possible welfare.
    \item The bids form an equilibrium, i.e., for each bidder. (Def.~\ref{def:equilibrium}):
    \begin{enumerate}
        \item Their tCPA constraint is satisfied at these bids (up to a factor of $1+\gamma$).
        \item They can not change their bid-multiplier and gain more than $\Delta$ additive value while staying within $(1+\gamma)$ of the tCPA constraint.
    \end{enumerate}
\end{enumerate}
\paragraph{Proving the result for any randomized auction.} The challenge is to prove the above for any randomized auction. To do this, firstly, note that our instance is almost universal, i.e., it only depends on the given auction via the max-probability $\pi^*$ and up to scaling of all numbers in the instance by the max-threshold $M^*$. Secondly, we rely only on basic properties that all anonymous randomized truthful auctions satisfy: (a) Lemma~\ref{lem:1cs-cost} which tackles the dependency of $\pi^*$ and $M^*$, (b) consequences of anonymity such as Lemma~\ref{lem:fairness-fraction}, and (c) simple properties following from truthfulness.\medskip

{\bf (1) Welfare: } 
With these bids, the auction for $P_i$ has only one bidder $A_i$ and its bid is at least the max-threshold (which is at most 1), 
so from Def.~\ref{def:max-threshold} $P_i$ is allocated to $A_i$ w.p. $\pi^*$. Query $Q_i$ gets allocated to the ``correct" bidder $B_i$ with probability at most $\frac{1}{k+1}$ since it is the lowest among $k+1$ bidders (Lemma~\ref{lem:fairness-fraction}). 
Thus the expected value obtained for the $P_j$ is $a\cdot V$ each, and for each $Q_j$ is at most $\frac{1}{k+1} V +  \epsilon (a\cdot V +k\rho)$ (the latter term is an upper bound, being the highest possible value among the $A_j$ for $Q_i$). 

Thus the total value achieved in the auctions under these bids is (ignoring terms with $\rho$)\footnote{The constant $\rho$ can be taken to be very small, in fact it can be considered to be a formal variable, introduced purely for tie-breaking among the $A_i$. Thus, we will not include the $\rho$ in the quantification of welfare and incentives.} 
\begin{align}
\cw(\eq) &\leq \pi^*\cdot \frac{a\cdot V}{\pi^*}\cdot k + \frac{k}{k+1} V + \epsilon \cdot a\cdot V\cdot k \notag \\
&= \left( a + \frac{1}{k+1} + \epsilon \cdot a\right)\cdot V\cdot k. \label{eq:welfare-in-k-eq}
\end{align}

The optimal allocation allocates each $P_i$ to $A_i$, and $Q_i$ to $B_i$, giving a value of 
\begin{equation}
    \label{eq:opt-value-k-bidders}
    OPT = \left(\frac{a}{\pi^*} + 1\right)\cdot V\cdot k.
\end{equation}

{\bf (2) Proving equilibrium: }
It remains to prove that this set of bids constitutes a $(\Delta, \gamma)-$equilibrium.

{\bf (2a) Bidders $A_i$'s constraint: } 
From Lemma~\ref{lem:1cs-cost}: $A_i$ is allocated $P_i$ w.p. $\pi^*$ and the cost for $P_i$ is at most $\pi^*$ since we re-scaled to set the max-threshold $M^* \leq 1$. Also, the per-unit price for $A_i$ for any $Q_j$ is at most its per-unit bid for $Q_j$. So we get (dropping the terms involving $\rho$):
\begin{align}
value(A_i) &= v(A_i, P_i) \cdot P[P_i \mbox{ allocated to }A_i]\notag\\ 
&+ \sum_{j=0}^{k-1} v(A_i, Q_j) \cdot P[Q_j \mbox{ allocated to }A_i] 
\notag\\    
&\geq a\cdot V + \sum_{j=0}^{k-1}\epsilon \cdot a\cdot V \cdot P[Q_j \mbox{ allocated to }A_i] 
\label{eq:value-ai}
\end{align}

\begin{align}
cost(A_i) &\leq \pi^* + \sum_{j=0}^{k-1} bid(A_i, Q_j) \cdot P[Q_j \mbox{ allocated to }A_i] \notag\\
&\leq 1 + \sum_{j=0}^{k-1} a\cdot V \cdot P[Q_j \mbox{ allocated to }A_i]\label{eq:cost-ai}
\end{align}

To get a handle on the probabilities in Eqs.~\ref{eq:value-ai} and~\ref{eq:cost-ai}, we consider {\bf a special setting} in which $\mathcal{A}$ is given $k+1$ bidders with bids equal to the tuple $$\beta = (V, \ a\cdot V+\rho,\  a\cdot V+2\rho\ , ..., a\cdot V+k\rho)$$
Let $(\chi_0, \chi_1, ..., \chi_k)$ be the probabilities with which $\mathcal{A}$ allocates the item to bidders $0,.., k$ in this setting. Clearly, $\sum_{s=0}^k \chi_s \leq 1$.

Given this property of $\mathcal{A}$, we can bound the total probability with which the $Q_j$ are allocated in our instance. Note that each $Q_j$ has $k+1$ bidders with a set of bids equal to the tuple $\beta$ defined above. Hence each $Q_j$ will be allocated according to the allocation $\{ \chi\}$ defined above. Now due to the rotational symmetry across the $Q_j$, $A_i$ is the $t_{ij} := ((i+j) \mod k +1)^\emph{th}$ lowest bidder for $Q_j$ among the $\{A_{\cdot}\}$ (while $B_j$ is the lowest bidder for $Q_j$). Therefore, as a consequence of anonymity, $A_i$ gets allocated $Q_j$ with probability $\chi_{t_{ij}}$. As $j$ ranges over $[0, k-1]$, $t_{ij}$ ranges over $[1, k]$. Thus, we get:
\begin{equation}
\label{eq:total-q-to-a}
    \sum_{j=0}^{k-1} P[Q_j \mbox{ allocated to }A_i] = \sum_{s=1}^{k} \chi_s \leq 1
\end{equation}
That is, for all $i$, $A_i$ is allocated an at most a unit expected total amount of any of the $Q_j$.

Thus, using Eq.~\ref{eq:total-q-to-a} in Eqs.~\ref{eq:value-ai} and~\ref{eq:cost-ai}, we get:

\begin{equation}
    \label{ai-gamma-constraint}
    cost(A_i) \leq \left(1 + \frac{1}{a\cdot V}\right) value(A_i)
\end{equation}

{\bf (2b) Bounding the gain from defection for bidder $A_i$: }
Reducing the bid multiplier can not help since that can only reduce the value obtained. Now even if $A_i$ increases its bid to infinity, the maximum value it can achieve is (ignoring the terms with $\rho$):
\begin{align*}
\pi^* \cdot v(A_i, P_i) + \sum_{j=0}^{k-1} v(A_i, Q_j) 
\leq a\cdot V + k \cdot\epsilon \cdot a\cdot V 
\end{align*}
Thus, using Eq.~\ref{eq:value-ai}, the highest that $A_i$ can gain in value by deviating from $\mu(A_i)$ (even ignoring its tCPA constraint) is
\begin{equation}
\label{eq:ai-gain-bound}
Gain(A_i) \leq k\cdot\epsilon\cdot a\cdot V
\end{equation}

\indent{\bf (2a) Bidders $B_i$'s constraint: } Firstly, for each $i$, $B_i$ only bids on one query ($Q_i$), and its bid is equal to its value, since $\mu(B_i) = 1$. Since the per-unit cost is at most the bid, $B_i$'s tCPA constraint is satisfied no matter the probability with which $Q_i$ is allocated to $B_i$.

{\bf (2b) Bounding the gain from defection for bidder $B_i$: }
To understand the incentive to deviate from $\mu(B_i) = 1$, note first that lowering its bid can only decrease its value obtained. Since $B_i$ is the lowest bidder for $Q_i$, and the other bids are $a + j \rho$ for $j \in [1, k]$, by Lemma~\ref{lem:fairness-fraction} we have that $B_i$ is allocated $Q_i$ with probability at most $\frac{1}{k+1}$. Now suppose $B_i$ raises its bid-multiplier so that the bid for $Q_i$ increases to some value $bid'(B_i, Q_i)$ and it wins $Q_i$ with some increased probability $p$.

Since the auction is truthful, the total cost to $B_i$ at a bid of $bid'(B_i, Q_i)$ can computed using Myerson's lemma for truthful pricing (the area above the allocation curve). Consider the function $\mathcal{F}$ which denotes the probability of allocating to $B_i$ according to $\mathcal{A}$ as $B_i$ changes its bid, given that the other bidders $\{A_j\}$ bids are fixed at $\{aV + j\rho\}_{j \in [1, k]}$. 
Consider the curve at the bid $bid'(B_i, Q_i)$ giving a probability of $p$. 
The value to $B_i$ is 
$$value'(B_i) = p\cdot V$$
The cost to $B_i$ is
\begin{align*}
cost'(B_i) &= \int_{b=0}^{bid'(B_i, Q_i)} \Big(p - \mathcal{F}(b)\Big) db\\
&\geq \Big(p - \mathcal{F}(a\cdot V)\Big) \cdot a\cdot V \\
& \geq \left(p - \frac{1}{k+1}\right) \cdot a\cdot V
\end{align*}
Here, first inequality follows because the allocation curve is non-decreasing, and $bid'(B_i, Q_i) \geq a\cdot V$, since the lowest bid among the other bidders is $a\cdot V+\rho$. The second inequality follows from Lemma~\ref{lem:fairness-fraction}, because at a bid of $a \cdot V$, $B_i$ is still the lowest among $k+1$ bidders. Fig.~\ref{fig:bi-qi-defection-cost-myerson} gives a pictorial representation of this bound on the cost.
\begin{figure}[t]
	\centering
  	\includegraphics[width=0.75\textwidth]{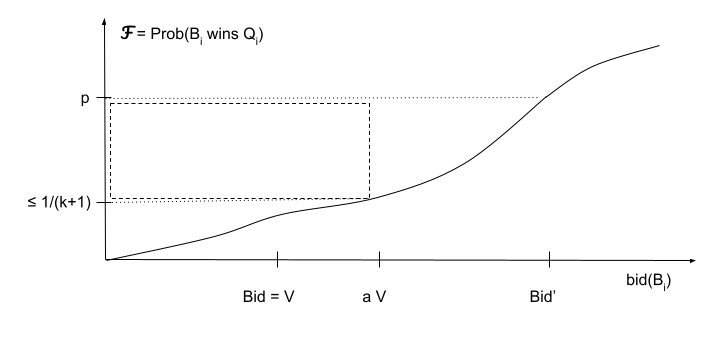}
  		\caption{Bounding the cost for $B_i$ as it raises its bid by the area of the rectangle in the figure.}
  	\label{fig:bi-qi-defection-cost-myerson}
\end{figure}

Since the tCPA constraint needs to be satisfied at the new bid up to a factor of $(1+\gamma)$, we have (dropping the terms with $\rho$):
\begin{align*}
cost'(B_i) \leq (1+\gamma)\cdot value'(B_i) 
 ~~&\Rightarrow \left(p - \frac{1}{k+1}\right) \cdot a\cdot V \leq (1+\gamma)\cdot p \cdot V \\
 & \Rightarrow p \leq \frac{a}{(a- 1 - \gamma)(k+1)}
\end{align*}
Thus we have a bound on the maximum possible gain that $B_i$ can get by changing its bid while still respecting its tCPA constraint:
\begin{equation}
    \label{eq:bi-gain-bound}
    Gain(B_i) \leq \frac{a\cdot V}{(a- 1-\gamma)(k+1)}
\end{equation}

{\bf Setting the parameters.} Now we can set our parameters to satisfy all the constraints for a $(\Delta, \gamma)-$equilibrium with a welfare close to a half of the optimal welfare.
From Eq.~\ref{ai-gamma-constraint}, we see that by taking $V > \frac{1}{\gamma}$, and noting that $a\geq 1$, $A_i$'s tCPA constraint is satisfied up to a $(1 + \gamma)$ factor (the tCPA constraint for each $B_i$ is strictly satisfied).
From Eq.~\ref{eq:bi-gain-bound}, we see that by taking $$k + 1 > \frac{aV}{(a-1-\gamma)\Delta} > \frac{a}{(a-1-\gamma)\Delta\gamma}$$
we get that $Gain(B_i)$ is no more than $\Delta$.  
From Eq.~\ref{eq:ai-gain-bound}, we see that by taking $\epsilon \rightarrow 0$, small enough so that $k\epsilon a V < \Delta$, we get  that $Gain(A_i)$ is no more than $\Delta$.
Finally, by taking $a \rightarrow 1$ (recall $a \geq 1$), we get from Eq.~\ref{eq:welfare-in-k-eq} and ~\ref{eq:opt-value-k-bidders} that the value obtained via the bidding in any randomized auction is at most
$$\frac{a}{\tfrac{a}{\pi^*}+1} + \frac{1}{(k+1)(\tfrac{a}{\pi^*}+1)} +\frac{\epsilon a}{\tfrac{a}{\pi^*}+1} \rightarrow \frac{\pi^*}{1 + \pi^*} \leq \frac{1}{2},$$
(as $a\rightarrow 1, \epsilon\rightarrow 0$ and $k\rightarrow\infty$), thus proving the theorem.
\end{proof}

\bibliographystyle{alpha}
\bibliography{poa}

\newcommand{\etalchar}[1]{$^{#1}$}
\begin{thebibliography}{DMMZ21}

\bibitem[AB20]{AllouahB20}
Amine Allouah and Omar Besbes.
\newblock Prior-independent optimal auctions.
\newblock {\em Manag. Sci.}, 66(10):4417--4432, 2020.

\bibitem[ABM19]{AggarwalBM19}
Gagan Aggarwal, Ashwinkumar Badanidiyuru, and Aranyak Mehta.
\newblock Autobidding with constraints.
\newblock In {\em Web and Internet Economics - 15th International Conference,
  {WINE} 2019, Proceedings}, volume 11920 of {\em Lecture Notes in Computer
  Science}, pages 17--30. Springer, 2019.

\bibitem[AFGR17]{AzarFGR17}
Yossi Azar, Michal Feldman, Nick Gravin, and Alan Roytman.
\newblock Liquid price of anarchy.
\newblock In {\em Algorithmic Game Theory - 10th International Symposium,
  {SAGT} 2017, L'Aquila, Italy, September 12-14, 2017, Proceedings}, pages
  3--15, 2017.

\bibitem[BDM{\etalchar{+}}21]{BalseiroDMMZ21}
Santiago~R. Balseiro, Yuan Deng, Jieming Mao, Vahab~S. Mirrokni, and Song Zuo.
\newblock The landscape of auto-bidding auctions: Value versus utility
  maximization.
\newblock In {\em {EC} '21: The 22nd {ACM} Conference on Economics and
  Computation, 2021}, pages 132--133. {ACM}, 2021.

\bibitem[BG19]{BalseiroG19}
Santiago~R. Balseiro and Yonatan Gur.
\newblock {Learning in Repeated Auctions with Budgets: Regret Minimization and
  Equilibrium}.
\newblock {\em Management Science}, 65(9):3952--3968, September 2019.

\bibitem[DL14]{DobzinskiL14}
Shahar Dobzinski and Renato~Paes Leme.
\newblock Efficiency guarantees in auctions with budgets.
\newblock In {\em Automata, Languages, and Programming - 41st International
  Colloquium, {ICALP} 2014, Copenhagen, Denmark, July 8-11, 2014, Proceedings,
  Part {I}}, pages 392--404, 2014.

\bibitem[DMMZ21]{DengMMZ21}
Yuan Deng, Jieming Mao, Vahab Mirrokni, and Song Zuo.
\newblock Towards efficient auctions in an auto-bidding world.
\newblock In {\em Proceedings of the Web Conference 2021}, WWW '21, page
  3965–3973. Association for Computing Machinery, 2021.

\bibitem[DRY15]{DhangwatnotaiRY15}
Peerapong Dhangwatnotai, Tim Roughgarden, and Qiqi Yan.
\newblock Revenue maximization with a single sample.
\newblock {\em Games and Economic Behavior}, 91:318--333, 2015.

\bibitem[Fac22]{fbautobiddingsupport}
Auto-bidding~products~support~page.
\newblock \url{https://www.facebook.com/business/help/1619591734742116}, 2022.
\newblock Accessed: 2022-02-09.

\bibitem[FILS15]{FuILS15}
Hu~Fu, Nicole Immorlica, Brendan Lucier, and Philipp Strack.
\newblock Randomization beats second price as a prior-independent auction.
\newblock In {\em Proceedings of the Sixteenth {ACM} Conference on Economics
  and Computation, {EC} '15, 2015}, page 323. {ACM}, 2015.

\bibitem[FMPS07]{FeldmanMPS07}
Jon Feldman, Shanmugavelayutham Muthukrishnan, Martin Pal, and Cliff Stein.
\newblock Budget optimization in search-based advertising auctions.
\newblock In {\em Proceedings of the 8th ACM conference on Electronic
  commerce}, pages 40--49, 2007.

\bibitem[GLPL21]{GolrezaiLP21}
Negin Golrezaei, Ilan Lobel, and Renato Paes~Leme.
\newblock Auction design for roi-constrained buyers.
\newblock In {\em Proceedings of the Web Conference 2021}, WWW '21, page
  3941–3952, 2021.

\bibitem[Goo22]{googleautobiddingsupport}
Auto-bidding~products~support~page.
\newblock \url{https://support.google.com/google-ads/answer/2979071}, 2022.
\newblock Accessed: 2022-02-09.

\bibitem[HR14]{hartline2014optimal}
Jason~D. Hartline and Tim Roughgarden.
\newblock Optimal platform design, 2014.

\bibitem[Mye81]{Myerson81}
Roger Myerson.
\newblock Optimal auction design.
\newblock {\em Mathematics of Operations Research}, 6(1):58--73, 1981.

\end{thebibliography}

\end{document}